\documentclass[11pt, onecolumn]{IEEEtran}

\usepackage{amsmath}
\usepackage{amssymb}
\usepackage{amsfonts}
\usepackage{amsthm}
\usepackage{epsfig}
\usepackage{subfigure}
\usepackage{calc}
\usepackage{color}
\usepackage[all]{xy}
\usepackage{bm}
\usepackage{algorithmicx}
\usepackage{algpseudocode}
\usepackage{algorithm}
\usepackage{enumerate}
\usepackage{hyperref}
\usepackage{multirow}
\usepackage{hhline}

\newtheorem{defn}{Definition}

\newtheorem{thm}{Theorem}[section]
\newtheorem{cor}[thm]{Corollary}
\newtheorem{prop}{Proposition}

\newtheorem{lem}[thm]{Lemma}
\newtheorem{conj}[thm]{Conjecture}
\newtheorem{constr}[thm]{Construction}
\newtheorem{note}{Remark}
\newtheorem{example}{Example}
\newcommand{\bit}{\begin{itemize}}
\newcommand{\eit}{\end{itemize}}
\newcommand{\bcor}{\begin{cor}}
\newcommand{\ecor}{\end{cor}}
\newcommand{\beq}{\begin{equation}}
\newcommand{\eeq}{\end{equation}}
\newcommand{\beqn}{\begin{equation*}}
\newcommand{\eeqn}{\end{equation*}}
\newcommand{\bea}{\begin{eqnarray}}
\newcommand{\eea}{\end{eqnarray}}
\newcommand{\bean}{\begin{eqnarray*}}
\newcommand{\eean}{\end{eqnarray*}}
\newcommand{\ben}{\begin{enumerate}}
\newcommand{\een}{\end{enumerate}}
\newcommand{\bdefn}{\begin{defn}}
\newcommand{\edefn}{\end{defn}}
\newcommand{\bnote}{\begin{note}}
\newcommand{\enote}{\end{note}}
\newcommand{\bprop}{\begin{prop}}
\newcommand{\eprop}{\end{prop}}
\newcommand{\blem}{\begin{lem}}
\newcommand{\elem}{\end{lem}}
\newcommand{\bthm}{\begin{thm}}
\newcommand{\ethm}{\end{thm}}
\newcommand{\bconj}{\begin{conj}}
\newcommand{\econj}{\end{conj}}
\newcommand{\bconstr}{\begin{constr}}
\newcommand{\econstr}{\end{constr}}
\newcommand{\bpf}{\begin{proof}}
\newcommand{\epf}{\end{proof}}

\title{Weight Enumerators and Higher Support Weights of Maximally Recoverable Codes}
\author{V. Lalitha and Satyanarayana V. Lokam
\thanks{V. Lalitha is with Signal Processing and Communications Research Center, International Institute of Information Technology, Hyderabad,
 India (email:  lalitha.v@iiit.ac.in). This work was carried out when the author was doing an internship with Microsoft Research, Bangalore.}
\thanks{Satyanarayana V. Lokam is with Microsoft Research, Bangalore, India (email:  satya@microsoft.com).} }

\begin{document}

\maketitle

\begin{abstract}
 In this paper, we establish the matroid structures corresponding to data-local and local maximally recoverable codes (MRC). The matroid structures of these codes can be used to determine the associated Tutte polynomial. Greene proved that the weight enumerators of any code can be determined from its associated Tutte polynomial. We will use this result to derive explicit expressions for the weight enumerators of data-local and local MRC. Also, Britz proved that the higher support weights of any code can be determined from its associated Tutte polynomial. We will use this result to derive expressions for the higher support weights of data-local and local MRC with two local codes.
\end{abstract}

\section{Introduction}\label{sec:intro}

In a distributed storage system, efficient repair of failed nodes is becoming increasingly important in addition to ensuring a given level of reliability and low storage overhead. Two recent approaches for efficient node repair are regenerating codes and codes with locality. Regenerating codes, introduced in \cite{DimGodWuWaiRam}, tradeoff repair bandwidth for storage overhead. On the other hand, codes with locality \cite{GopHuaSimYek} tradeoff repair degree (number of nodes accessed to repair a failed node) for storage overhead. In this paper, we will deal with codes with locality.

Let $\mathcal{C}$ be an $[n,k,d_{\min}]$ linear code $\mathcal{C}$ over the field $\mathbb{F}_q$. The $i$th code-symbol $c_i$, $1 \leq i
\leq n$,  is said to have locality $r$ if this symbol can be recovered by accessing at most $r$ other code symbols of code $\mathcal{C}$. A systematic linear code $\mathcal{C}$ having minimum distance $d$, where all $k$ message
symbols have locality $r$, is said to be an $(r, d_{\min})$ code. It has been proved in \cite{GopHuaSimYek} that the minimum distance of an $(r,d_{\min})$ code is upper bounded by
\bea \label{eq:gopalan_bound}
d_{\min} & \leq &  n- k -\left\lceil \frac{k}{r} \right\rceil + 2.
\eea

A code is said to have information locality if the $k$ message symbols have locality $r$ and it is said to have all-symbol locality if all the $n$ code symbols have locality $r$. Pyramid codes, which were constructed earlier \cite{HuaCheLi}, are shown to be optimal codes with information locality.  For the case when $(r+1) \mid n$, the existence of optimal codes with all-symbol locality over a large field size was established. Families of codes with all-symbol locality over low field size are constructed in \cite{TamBar}. Locality in the setting of nonlinear codes has been studied in  \cite{PapDim}, \cite{ForYek}.  The connection between codes with locality (also known as locally recoverable codes (LRC)) and matroids has been studied in \cite{TamPapDim}. Codes with locality have been implemented and their performance evaluated in two systems, the first is Windows Azure storage  \cite{HuaSimXu_etal_azure} and the second is Hadoop Distributed File System \cite{HuaSimXu_etal_azure, sathiamoorthy}.

\subsection{Maximally Recoverable Codes}
The idea of maximal recoverability of a code was introduced in \cite{CheHuaLi}. Maximally recoverable codes in the context of codes with locality have been studied in \cite{GopHuaJenYek}, \cite{BlaHafHet}.
A code is said to be maximally recoverable if it corrects all erasure patterns which can potentially be corrected, given the locality constraints.

Let $\mathcal{C}$ denote an $[n, k, d_{\min}]$ linear code over $\mathbb{F}_q$ in the systematic form. Let $\mathcal{I} = [k]$
denote the indices of the message symbols. For any set $S \subseteq [n]$, we will use $\mathcal{C}|_S$ to denote the
restriction of $\mathcal{C}$ to the coordinates indexed by $S$.  

\vspace{0.1in}

\begin{defn}[Data-Local Maximally Recoverable Code]
Consider an $[n,k,d_{\min}]$ code $\mathcal{C}$ over $\mathbb{F}_q$ with information-symbol locality $r$, where $n = k + \frac{k}{r} + h$. Let a codeword of $\mathcal{C}$ be denoted by $\bold{c} = (c_1, \ldots, c_n)$, where $c_1, \ldots, c_k$ denote the message symbols and $c_{k+1}, \ldots, c_n$ denote the parity symbols. Let $c_{k+i}, 1 \leq i \leq \frac{k}{r}$ denote the local parity corresponding to the message symbols $(c_{(i-1)r+1}, \ldots c_{ir})$, i.e., $c_{k+i} = \sum_{j=1}^r a_{i,j} \ c_{(i-1)r+j}$, where all $a_{i,j} \in \mathbb{F}_q$ are nonzero. The remaining $n-k-\frac{k}{r} = h$ parities corresponding to parity symbols $c_{k+\frac{k}{r} + 1}, \ldots, c_n$ are global parities and are allowed to depend on all $k$ message symbols. Let the supports of the $\frac{k}{r}$ local groups be $\{ S_1, S_2, \ldots, S_{\frac{k}{r}} \}$. $\mathcal{C}$ is said to be a maximally recoverable code if for any set $E$ such that $|E| = k+h$ and $|E \cap S_i| = r, 1 \leq i \leq \frac{k}{r}$, the punctured code $\mathcal{C}|_E$ is a $[k+h,k,h+1]$ MDS code.
\end{defn}

\vspace{0.1in}

\begin{defn}[Local Maximally Recoverable Code]

Consider an $[n,k,d_{\min}]$ code $\mathcal{C}$ over $\mathbb{F}_q$ with all-symbol locality $r$, where $\frac{n}{r+1} = \frac{k +  h}{r}$.  Let a codeword of $\mathcal{C}$ be denoted by $\bold{c} = (c_1, \ldots, c_n)$, where $c_1, \ldots, c_k$ denote the message symbols and $c_{k+1}, \ldots, c_n$ denote the parity symbols. The first $h$ parities corresponding to parity symbols $c_{k + 1}, \ldots, c_{k+h}$ are global parities and are allowed to depend on all $k$ message symbols. Let $c_{k+h+i}, 1 \leq i \leq \frac{k+h}{r}$ denote the local parity corresponding to the code symbols $(c_{(i-1)r+1}, \ldots c_{ir})$, i.e., $c_{k+h+i} = \sum_{j=1}^r a_{i,j} \ c_{(i-1)r+j}$, where all $a_{i,j} \in \mathbb{F}_q$ are nonzero. Let the supports of the $\frac{k+h}{r}$ local groups be $\{ S_1, S_2, \ldots, S_{\frac{k+h}{r}} \}$. $\mathcal{C}$ is said to be a maximally recoverable code if for any set $E$ such that $|E| = k+h$ and $|E \cap S_i| = r, 1 \leq i \leq \frac{k+h}{r}$, the punctured code $\mathcal{C}|_E$ is a $[k+h,k,h+1]$ MDS code.

\end{defn}

In \cite{GopHuaJenYek}, explicit data-local and local MRC are constructed over a field size of the order of $k^{h-1}$. In \cite{BlaHafHet}, low field-size constructions of local MRC are given for specific set of parameters. In \cite{BlaHafHet}, the MRC are designed for application in an SSD (solid-state devices) setting, where there can be combination of sector and disk failures. For the same application, partial MRCs where specific patterns of erasures can be corrected, given the locality constraints, have been constructed in \cite{PlaBlaHaf}, \cite{LiLee}, \cite{BlaPlaSchYaa} for specific sets of parameters. 

We note here that if a code is either data-local or local maximally recoverable, then it is optimal with respect to the bound in \eqref{eq:gopalan_bound}. A different notion of partial MRC has been introduced in \cite{BalKum}, where it is required that the code is optimal and it is MDS when punctured on a specific set of coordinates.
\vspace{0.1in}

\subsection{Overview of Results} \label{sec:overview_results}

In this paper, we identify the matroid structures corresponding to $(k, r, h)$ data-local and local maximally recoverable codes (MRC). We note here that this is the main property differentiating an MRC from a locally recoverable code (LRC). Given the parameters $(k, r, h)$, the matroid structure of an MRC is fixed but the same is not true for LRC. The matroid structure of MRC is used to determine its associated Tutte polynomial.  
In \cite{Gre}, Greene proved that the weight enumerators of any code can be determined from its associated Tutte polynomial. We will apply the result in order to give explicit expressions for the weight enumerators of $(k, r, h)$ data-local MRC. These expressions are derived for a general class of parameters when the number of local codes $\ell < r+1$. The main calculation involved in deriving these expressions is to count the number of sets $U \subseteq [n]$ such that for a given $(u,v)$ $0 \leq u \leq n, 0 \leq v \leq k$, the size of $U$ is $u$ and rank of the code punctured to $U$ is $v$. Using the same techniques, we also derive the weight enumerators of local MRC with two local codes. In \cite{Bri_hsw}, \cite{Bri_critical}, Britz proved that the higher support weights of any code can be determined from its associated Tutte polynomial. 
We will apply the result in order to give  expressions for the higher support weights of data-local and local MRC with two local codes.

Section~\ref{sec:background} provides background on matroids, Tutte polynomial and the relation between weight enumerators and Tutte polynomial of a code.  The weight enumerators of data-local MRC are derived in Section \ref{sec:weights_datalocal_mrc} and the weight enumerators of local MRC with two local codes are derived in Section \ref{sec:weights_local_mrc}. Section \ref{sec:hsw_mrc} presents the higher support weights of data-local and local MRC with two local codes. Finally, we give conclusions and ongoing work in Section \ref{sec:concl}.

\vspace{0.1in}

\section{Matroids, Tutte Polynomial and Weight Enumerators} \label{sec:background}

In this section, we define a matroid and the matroid associated with a code. Then, we introduce the Tutte polynomial corresponding to a code.  We will present the result by Greene \cite{Gre}, which relates the Tutte polynomial of a code to its weight enumerator. We will illustrate all the ideas using the example of MDS codes and finally derive the weight enumerators of MDS codes.

\begin{defn}[Matroid]
A matroid $M$ is defined by a pair $M =  (S, \mathcal{I})$, where $S$ is a finite ground set and $\mathcal{I}$ is the set of subsets of $S$ which are termed independent sets.
The matroid $M$ is required to satisy the following axioms:
\ben
\item If $I \in \mathcal{I}$ and $J \in I$, then $J \in \mathcal{I}$.
\item If $I,J \in \mathcal{I}$ and $|J| > |I|$, then there exists an element $z \in J \setminus I$ such that $I \cup \{z\} \in \mathcal{I}$.
\een
\end{defn}
Rank function of a matroid $M$ maps an arbitrary subset $U$ of $S$ to the size of maximally independent set in $U$.
\begin{equation}
\text{rank}(U) = \max_{I \subseteq U, I \in \mathcal{I}} |I|.
\end{equation}

Let $\mathcal{C}$ be a code and $G$ denote its generator matrix.
Then, the matroid corresponding to the code $M(\mathcal{C})$ has the ground set as $S = \{1, \ldots, n\}$. The independent sets of the matroid are given by
\begin{equation}
\mathcal{I} = \{ I \subseteq S \ | \ \text{rank}(G|_I) = |I| \}.
\end{equation}
Rank function $\rho$ of the matroid $M(\mathcal{C})$ is given by $\rho(U) = \text{rank}(G|_U)$.

\begin{example}
Consider an $[n,k,n-k+1]$ MDS code. Since any $n-k$ erasures can be tolerated by the code, every subset of $S = [n]$ of size at most $k$ is an independent set.
\begin{equation}
\mathcal{I} = \{ I \subseteq S \ | \ |I| \leq k \}.
\end{equation}
The matroid defined above is termed as uniform matroid.
\end{example}

Consider a matroid $M(\mathcal{C})$ corresponding to a code $\mathcal{C}$ with parameters $[n,k,d_{\min}]$. The Tutte polynomial corresponding to the code $\mathcal{C}$ is a bivariate polynomial given by
\begin{equation} \label{eq:tutte}
T_C(X,Y) = \sum_{U \subseteq [n]} (X-1)^{k - \rho(U)}(Y-1)^{|U| - \rho(U)}.
\end{equation}
The Tutte polynomial depends only the matroid structure of the code. Hence, it is termed as {\em matroid invariant}. It can be observed from \eqref{eq:tutte} that to determine the Tutte polynomial of a code, for any given pair $(u, v)$ ($0 \leq u \leq n$ and $0 \leq v \leq k$), we have to count the number of sets $U \subseteq S$ such that $|U| = u$ and $\rho(U) = v$.

\begin{example} \label{ex:tutte_mds}
Consider an $[n,k,n-k+1]$ MDS code and its associated uniform matroid $M(\mathcal{C})$. We will evaluate the Tutte polynomial for the MDS code. It can be seen that for any set $U$, 
\ben
\item If $|U| = u \leq k$, then $\rho(U) = u$. The number of sets in this case is ${n \choose u}$.
\item If $|U| = u > k$, then $\rho(U) = k$. Even in this case, the number of sets is  ${n \choose u}$.
\een
The Tutte polynomial associated with the MDS code is given by
\begin{eqnarray} \label{eq:tutte_mds2}
T_C(X,Y) & = &  \sum_{U \subseteq [n]} (X-1)^{k - \rho(U)}(Y-1)^{|U| - \rho(U)} \nonumber \\
& = & \sum_{u=0}^k {n \choose u} (X-1)^{k-u} + \sum_{u=k+1}^n {n \choose u} (Y-1)^{u-k}.
\end{eqnarray}
\end{example}

The weight enumerator of a code $\mathcal{C}$ is defined as a polynomial $W_C(Z) = \sum_{i=0}^n A_i Z^i$, \\
where $A_i = | \bold{c} \in \mathcal{C} : wt(\bold{c}) = i|$. The weight enumerator of a code is related to the Tutte polynomial of the code as follows (Greene, \cite{Gre}):

\begin{thm}
The weight enumerator of an $[n,k,d_{\min}]$ code $W_C(Z)$ can de determined from the Tutte polynomial of the code $T_C(X,Y)$ using the following relation:
\begin{equation}
W_C(Z) = Z^{n-k}(1-Z)^k T_C( \frac{1+(q-1)Z}{1-Z}, \frac{1}{Z} ).
\end{equation}
\end{thm}
The above equation can be simplified as follows.
\begin{eqnarray}
W_C(Z) & = & Z^{n-k}(1-Z)^k T_C( \frac{1+(q-1)Z}{1-Z}, \frac{1}{Z} ) \nonumber \\
& = &      Z^{n-k}(1-Z)^k      \sum_{U \subseteq [n]} \left ( \frac{1+(q-1)Z}{1-Z} -1 \right )^{k - \rho(U)}\left (\frac{1}{Z}-1 \right )^{|U| - \rho(U)} \nonumber \\
& = &  Z^{n-k}(1-Z)^k      \sum_{U \subseteq [n]} \left ( \frac{qZ}{1-Z} \right )^ {k - \rho(U)}  \left (\frac{1}{Z}-1 \right )^{|U| - \rho(U)} \nonumber \\
& = &  \sum_{U \subseteq [n]} Z^{n- |U|} (1-Z)^{|U|} q^{k-\rho(U)}. \label{eq:weight_set_sum}
\end{eqnarray}
Since the minimum distance of the code is $d_{\min}$, it is clear that weight enumerator polynomial has to be such that $A_0 = 1, A_i = 0, 1 \leq i \leq d_{\min}-1$. To ensure this, we will split the above expression into two parts as follows:
\begin{eqnarray}
W_C(Z) & = &  \sum_{U \subseteq [n]} Z^{n- |U|} (1-Z)^{|U|} q^{k-\rho(U)} \nonumber \\
& = & \sum_{\substack{U \subseteq [n]: \\ |U| \leq n-d_{\min}}} Z^{n- |U|} (1-Z)^{|U|} q^{k-\rho(U)} + \underbrace{\sum_{\substack{U \subseteq [n]: \\ |U| \geq n-d_{\min}+1}} Z^{n- |U|} (1-Z)^{|U|} q^{k-\rho(U)}}_{T}. \label{eq:weight_poly1}
\end{eqnarray}
We note that for any set $U$ such that $|U| \geq n-d_{\min}+1$, $\rho(U) = k$. Hence, the second term $T$ in the \eqref{eq:weight_poly1} is given by
\begin{eqnarray}
T & = & \sum_{\substack{U \subseteq [n]: \\ |U| \geq n-d_{\min}+1}} Z^{n- |U|} (1-Z)^{|U|} \nonumber \\
& = & \sum_{u = n-d_{\min}+1}^n  {n \choose u} Z^{n- u} (1-Z)^{u} \nonumber \\
& = & 1 - \sum_{u = 0}^{n - d_{\min}}  {n \choose u} Z^{n- u} (1-Z)^{u} \nonumber \\
& = & 1 - \sum_{\substack{U \subseteq [n]: \\ |U| \leq n-d_{\min}}} Z^{n- |U|} (1-Z)^{|U|},
\end{eqnarray}
%where $(a)$ is true since
%\begin{eqnarray}
%1 =  (Z + (1-Z))^n & = & \sum_{u = 0}^n  {n \choose u} Z^{n- u} (1-Z)^{u} \nonumber \\
%& = & \sum_{u = 0}^{n-d_{\min}}  {n \choose u} Z^{n- u} (1-Z)^{u} + \sum_{u = n-d_{\min}+1}^n  {n \choose u} Z^{n- u} (1-Z)^{u}.
%\end{eqnarray}
Thus, the weight enumerator in \eqref{eq:weight_poly1} can be rewritten as
\begin{equation}  \label{eq:weight_dmin}
W_C(Z) = 1 + \sum_{\substack{U \subseteq [n]: \\ |U| \leq n-d_{\min}}} Z^{n- |U|} (1-Z)^{|U|} (q^{k-\rho(U)} - 1).
\end{equation}

\begin{example}
Consider an $[n,k,d_{\min} = n-k+1]$ MDS code. We will apply the count of the pairs $(|U|, \rho(U))$ discussed in Example \ref{ex:tutte_mds} for calculating the weight distribution of the MDS code.
%\ben
%\item If $|U| = u \leq k$, then $\rho(U) = u$. The number of sets in this case is ${n \choose u}$.
%\item If $|U| = u > k$, then $\rho(U) = k$. Even in this case, the number of sets is  ${n \choose u}$.
%\een
 The weight enumerator polynomial of the MDS code is given by
 \begin{eqnarray}
 W_C(Z) & = & 1 + \sum_{\substack{U \subseteq [n]: \\ |U| \leq n-d_{\min}}} Z^{n- |U|} (1-Z)^{|U|} (q^{k-\rho(U)} - 1) \nonumber \\
 & = &  1 + \sum_{u = 0}^{k-1} {n \choose u} Z^{n-u} (1-Z)^u (q^{k-u} - 1) \nonumber \\
 & = & 1 + \sum_{u = 0}^{k-1} \sum_{j=0}^u {n \choose u} Z^{n-u} Z^j (-1)^j {u \choose j} (q^{k-u} - 1).
 \end{eqnarray}
 By applying the change of variables $w = n-u+j, j=j$ to the above equation, we have
 \begin{eqnarray}
  W_C(Z) & = & 1 + \sum_{w = d_{\min}}^{n} \sum_{j=0}^{w-d_{\min}} {n \choose n-w+j} Z^w (-1)^j {n-w+j \choose j} (q^{k+w-n-j} - 1) \nonumber \\
  & \stackrel{(a)}{=} & \sum_{w = d_{\min}}^{n} \sum_{j=0}^{w-d_{\min}} {n \choose w} {w \choose j} Z^w (-1)^j (q^{k+w-n-j} - 1),
 \end{eqnarray}
 where $(a)$ follows from the fact that ${n \choose n-w+j} {n-w+j \choose j} = {n \choose w} {w \choose j}$. Thus the weight enumerators for the MDS code are given by
 \begin{eqnarray}
 A_0 = 1, \ \ \ \ \  A_w & = & \sum_{j=0}^{w-d_{\min}} {n \choose w} {w \choose j}  (-1)^j (q^{k+w-n-j} - 1), \ \  d_{\min} \leq w \leq n.
 \end{eqnarray}
\end{example}

\section{Weight Enumerators of Data-Local MRC}  \label{sec:weights_datalocal_mrc}

In this section, we identify the matroid structure of a data-local MRC. We derive the weight enumerators of a data-local MRC with two local codes. Then, we apply the result to an example code which is employed in Windows Azure storage. Using Macwilliams identity, we also present the weight enumerators of the dual of data-local MRCs with two local codes. Finally, we consider the general case, when the number of local codes $\ell$ is such that $3 \leq \ell < r+1$ and derive the weight enumerators of the corresponding data-local MRCs.

\begin{prop} \label{prop:mrc_matroid}
Consider an $[n,k,d_{\min}]$ data-local MRC with locality $r$, where $n = k + \frac{k}{r} + h$ and $d_{\min} = h+2$. Let the supports of the $\frac{k}{r}$ local groups be $\{ S_1, S_2, \ldots, S_{\frac{k}{r}} \}$. The set of all independent sets of the data-local MRC are given by
\begin{equation}
\mathcal{I} = \{ I \subseteq S \ | \ |I| \leq k, |I \cap S_i| \leq r, 1 \leq i \leq \frac{k}{r}\}.
\end{equation}
\end{prop}
\bpf
If $I \nsubseteq \mathcal{I}$, then either $|I| \geq k+1$ or there exists a local group $i$ such that $|I \cap S_i| = r+1$. In both cases, the set $I$ is a dependent set either by the dimension of the code or by the locality of the code, respectively. If $I \in \mathcal{I}$, then $I$ can be extended to a set $E$ such that $|E \cap S_i| = r$. Hence, $\mathcal{C}|_E$ is a $[k+h, k, h+1]$ MDS code.  Since $I \subset E$ and $|I| \leq k$, $I$ is an independent set by the MDS property of $\mathcal{C}|_E$.
\epf

\subsection{Weight Enumerators of Data-Local MRC with Two Local Codes}

\begin{thm} \label{thm:weight_two_codes}
Consider an $[n,k,d_{\min}]$ data-local MRC with locality $r$, where $k = 2r$, $n = k + \frac{k}{r} + h$ and $d_{\min} = h+2$. Let the supports of the two local groups be $\{ S_1, S_2 \}$. The weight enumerators of the code are given by
\begin{eqnarray}
A_{h+2} & =  & 2 {n-r-1 \choose r-1}(q-1), \\
A_w & = &\sum_{j=0}^{w-h-3} {n \choose w}{w \choose j} (-1)^j (q^{k+w-n-j}-1) \nonumber \\
&&+ \sum_{j=0}^{w-h-2} 2 (q-1) {n-r-1 \choose n-w+j-r-1} {n-w+j \choose j} (-1)^j q^{k+w-n-j},  h+3 \leq w \leq r+h+1, \\
A_w & = &\sum_{j=0}^{w-h-3} {n \choose w}{w \choose j} (-1)^j (q^{k+w-n-j}-1) \nonumber \\
&&+ \sum_{j=w-r-h-1}^{w-h-2} 2 (q-1) {n-r-1 \choose n-w+j-r-1} {n-w+j \choose j} (-1)^j q^{k+w-n-j},  r+h+2 \leq w \leq n.
\end{eqnarray}

\end{thm}

\bpf We will evaluate the weight enumerator of the data-local MRC using the following expression:
\begin{equation} \label{eq:weight_poly}
W_C(Z) = 1 + \sum_{\substack{U \subseteq [n]: \\ |U| \leq n-d_{\min}}} Z^{n- |U|} (1-Z)^{|U|} (q^{k-\rho(U)} - 1).
\end{equation}
It can be observed from \eqref{eq:weight_poly} that to determine the weight enumerators of the code, we have to count the number of sets $U \subseteq S$ such that $|U| = u$ and $\rho(U) = v$ for any given pair $(u, v)$ ($0 \leq u \leq n-d_{\min}$ and $0 \leq v \leq k$). For the case of MDS code, the rank of a set is uniquely determined by the cardinality of the set. For the case of MRC, that is not true. Hence, the counting involves more cases based on the matroid structure of the code as given in Proposition \ref{prop:mrc_matroid} . For any set $U$ such that $|U| = u \leq n-d_{\min} = 2r$, we have the following cases
\vspace{0.1in}
\ben
\item For a set $U$ of size $0 \leq u \leq r$, the set is independent. Hence, rank of the set $\rho(U) = u$. The number of sets  of this type are $N_{u,0} = {n \choose u}$. 
\vspace{0.05in}
\item For a set $U$ of size $r+1 \leq u \leq 2r$, there are the following two cases possible.   
\ben
\item Let $U$ be such that either $|U \cap S_1| = r+1$ or $|U \cap S_2| = r+1$. In this case, the set $U$ is dependent. The rank of the set $\rho(U)$, which is the size of the largest independent set in $U$, is given by $\rho(U) = u-1$. The number of sets $N_{u,1}$ of this type are $N_{u,1} = 2{n-r-1 \choose u-r-1}$. This is because either of the two local codes can be chosen and the remaining $u-r-1$ coordinates can be chosen from $n-r-1$ coordinates in ${n-r-1 \choose u-r-1}$ ways.
\vspace{0.05in}
\item Let $U$ be such that $|U \cap S_i| \leq r, i = 1,2$. In this case, the set $U$ is independent and hence, rank of the set $\rho(U) = u$. The number of sets of this type are $N_{u,0} - N_{u,1} = {n \choose u} - 2{n-r-1 \choose u-r-1}$. 
\een
\een
Hence, the weight enumerator polynomial in \eqref{eq:weight_poly} can be rewritten as
\begin{eqnarray}
W_C(Z) &= &1+ \sum_{u=0}^r {n \choose u} Z^{n-u} (1-Z)^u (q^{k-u} - 1) \nonumber \\
&& + \sum_{u=r+1}^{2r} 2{n-r-1 \choose u-r-1} Z^{n-u} (1-Z)^u (q^{k-u+1} - 1) \nonumber \\
&& + \sum_{u=r+1}^{2r} \left ( {n \choose u} - 2{n-r-1 \choose u-r-1} \right ) Z^{n-u} (1-Z)^u (q^{k-u} - 1) \nonumber \\
& \stackrel{(a)}{=} & 1+  \sum_{u=0}^{2r-1} {n \choose u} Z^{n-u} (1-Z)^u (q^{k-u} - 1) \nonumber \\
&& + \sum_{u=r+1}^{2r} 2{n-r-1 \choose u-r-1} Z^{n-u} (1-Z)^u (q^{k-u+1} - q^{k-u}) \nonumber \\
& = & 1+ \sum_{u=0}^{2r-1} \sum_{j=0}^{u} {n \choose u} {u \choose j} Z^{n-u+j} (-1)^j (q^{k-u} - 1) \nonumber \\
&& + \sum_{u=r+1}^{2r} \sum_{j=0}^{u} 2{n-r-1 \choose u-r-1} {u \choose j} Z^{n-u+j} (-1)^j (q^{k-u+1} - q^{k-u}) \nonumber \\
& = & 1+ T_1 + T_2,
\end{eqnarray}
where $(a)$ follows since for the case when $u=2r=k$ and the set is full rank, the term in the summation is zero.
 By applying the change of variables $w = n-u+j, j=j$ to the above equation, we first identify the $(u,j)$ pairs which result in a given value of $w$ (Table \ref{tab:change_variable_two}).

\begin{table}[ht]
\centering
\begin{tabular}{||c|c|c|l|c||}
\hline
\hline
  Term & $w$ & No. of  & $(u,j)$ pairs & Range of $j$ \\
  in Sum & & terms & & \\
  \hline
        
 \multirow{4}{*}{$T_1$} & $n$ & $2r$ & $(0,0), (1,1), \ldots, (2r-1,2r-1)$ & \\
  \hhline{~---~}
  & $n-1$ & $2r-1$ & $(1,0), (2,1), \ldots, (2r-1, 2r-2)$ & $j=0$ to\\
  \hhline{~---~}
 & \vdots & \vdots & \vdots & $j=w-h-3$ \\
  \hhline{~---~}
& $h+3$ & $1$ & $(2r-1,0)$ & \\
\hline
\hline
  \multirow{7}{*}{$T_2$} & $n$ & $r$ & $(r+1,r+1), (r+2,r+2), \ldots, (2r,2r)$  &  \\
  \hhline{~---~}
  & $n-1$ & $r$ & $(r+1,r), (r+2,r+1), \ldots, (2r, 2r-1)$ &  $j=w-r-h-1$ \\
  \hhline{~---~}
  & \vdots & \vdots & \vdots & to $j=w-h-2$ \\
  \hhline{~---~}
  & $r+h+2$ & $r$ &  $(r+1,1), (r+2,r-1), \ldots, (2r,r)$ &  \\
  \hhline{~----}
& $r+h+1$ & $r$ &  $(r+1,0), (r+2,1), \ldots, (2r,r-1)$ & $j=0$ to \\
  \hhline{~---~}
 & \vdots & \vdots & \vdots & $j=w-h-2$ \\
  \hhline{~---~}
& $h+2$ & $1$ & $(2r,0)$ & \\
        \hline
      \hline

\end{tabular}
%\begin{tabular}{||c|c|l||c|c|l||} \hline
% \hline
%  \multicolumn{3}{||c||}{Second Term}  & \multicolumn{3}{||c||}{First Term}  \\
%  \hline
%  $w$ & No. of terms & $(u,j)$ pairs & $w$ & No. of terms & $(u,j)$ pairs \\
%  &  &  & & & \\
%  \hline
%  $n$ & $r$ & $(r+1,r+1), (r+2,r+2), \ldots, (2r,2r)$ &  $n$ & $2r$ & $(0,0), (1,1), \ldots, (2r-1,2r-1)$ \\
%  \hline
%  $n-1$ & $r$ & $(r+1,r), (r+2,r+1), \ldots, (2r, 2r-1)$ & $n-1$ & $2r-1$ & $(1,0), (2,1), \ldots, (2r-1, 2r-2)$ \\
%  \hline
%  \vdots & \vdots & \vdots & \vdots & \vdots & \vdots\\
%  \hline
%  $r+h+2$ & $r$ &  $(r+1,1), (r+2,r-1), \ldots, (2r,r)$ & \vdots & \vdots & \vdots \\
%  \hline
%$r+h+1$ & $r$ &  $(r+1,0), (r+2,1), \ldots, (2r,r-1)$ & \vdots & \vdots & \vdots\\
%  \hline
%  \vdots & \vdots & \vdots & \vdots & \vdots & \vdots \\
%  \hline
%$h+2$ & $1$ & $(2r,0)$ & $h+3$ & $1$ & $(2r-1,0)$\\
%        \hline
%      \hline
%\end{tabular}
\vspace{0.1in}
 \caption{Illustrating the change of variables from $u,j$ to $w=n-u+j, j$ for a data-local MRC with two local codes.} \label{tab:change_variable_two}

\end{table}

With reference to the table above, the weight enumerator polynomial can be rewritten in terms of variables $w,j$ as 
\begin{eqnarray}
W_C(Z) & = & 1+ \sum_{w=h+3}^{n} \sum_{j=0}^{w-h-3} {n \choose n-w+j} {n-w+j \choose j} Z^{w} (-1)^j (q^{k+w-n-j} - 1) \nonumber \\
    & & \hspace{-0.4in}+ \sum_{w=h+2}^{r+h+1} \sum_{j=0}^{w-h-2} 2{n-r-1 \choose n-w+j-r-1} {n-w+j \choose j} Z^{w} (-1)^j (q^{k+w-n-j+1} - q^{k+w-n-j}) \nonumber \\
    & & \hspace{-0.4in} + \sum_{w=r+h+2}^{n} \sum_{j=w-r-h-1}^{w-h-2} 2{n-r-1 \choose n-w+j-r-1} {n-w+j \choose j} Z^{w} (-1)^j (q^{k+w-n-j+1} - q^{k+w-n-j}).
\end{eqnarray}

The weight enumerators in the theorem statement follow by bringing together coefficients with the same power of $Z$ and by using the fact that ${n \choose n-w+j} {n-w+j \choose j} = {n \choose w} {w \choose j}$.
\epf

\begin{example}
Consider the local reconstruction code (LRC) discussed in \cite{HuaSimXu_etal_azure} with two local codes and two global parities. The parameters of the code are given by $k = 6, r = 3, h = 2$ and $n=10, d_{\min} = h+2 = 4$. The code employed in \cite{HuaSimXu_etal_azure} is an MRC over a finite field of size $q=16$. Thus, the weight enumerator polynomial of the code can be obtained by applying Theorem \ref{thm:weight_two_codes}.
\begin{equation}
W_C(Z) = 1 + 450 Z^4 + 3960 Z^5 + 34680 Z^6 + 304080Z^7 + 1782405Z^8 + 5847480Z^9 + 8804160 Z^{10}.
\end{equation}
\end{example}

The weight enumerator polynomial of the dual code $\mathcal{C}^\perp$ can be determined in terms of weight enumerator polynomial of the code $\mathcal{C}$ using MacWilliams identity \cite{MacSlo}.
\begin{eqnarray}
W_{C^\perp}(Z) & = & \frac{ (1 + (q-1)Z)^n}{|C|} W_C \left ( \frac{1-Z}{1+(q-1)Z} \right ) \nonumber \\
& \stackrel{(a)}{=} & \frac{ (1 + (q-1)Z)^n}{q^k} \sum_{U \subseteq [n]} \left (   \frac{1-Z}{1+(q-1)Z} \right )^{n-|U|} \left ( 1-\frac{1-Z}{1+(q-1)Z} \right )^{|U|} q^{k - \rho(U)} \nonumber \\
& = & \sum_{U \subseteq [n]}  Z^{|U|} (1-Z)^{n-|U|} q^{|U| - \rho(U)} \nonumber \\
& \stackrel{(b)}{=} & 1+ \sum_{\substack{U \subseteq [n]: \\ |U| \geq d_{\min}^{\perp}}}  Z^{|U|} (1-Z)^{n-|U|} (q^{|U| - \rho(U)}-1),
\end{eqnarray}
where $(a)$ follows from \eqref{eq:weight_set_sum} and $(b)$ can be shown similar to  \eqref{eq:weight_dmin}, since the minimum distance of the dual code is $d_{\min}^{\perp}$.

\begin{thm}
Consider an $[n,k,d_{\min}]$ data-local MRC $\mathcal{C}$ with locality $r$, where $k = 2r$, $n = k + \frac{k}{r} + h$ and $d_{\min} = h+2$. Let the supports of the two local groups be $\{ S_1, S_2 \}$. The weight enumerators of the dual code $\mathcal{C}^\perp$ are given by
\begin{eqnarray}
A_w & = &\sum_{j=0}^{w-r-1} 2 {n-w+j \choose j}{n-r-1 \choose w-j-r-1} (-1)^j (q-1) \ \ \ \ \ r+1 \leq w \leq 2r, \\
A_w & = &\sum_{j=0}^{w-2r-1} {n-w+j \choose j}{n \choose w-j} (-1)^j (q^{w-j-k}-1) \nonumber \\
&&+ \sum_{j=w-2r}^{w-r-1} 2 {n-r-1 \choose w-j-r-1} {n-w+j \choose j} (-1)^j (q-1), \ \ 2r+1 \leq w \leq n.
\end{eqnarray}
\end{thm}

\bpf
The proof follows from the same counting as in the proof of Theorem \ref{thm:weight_two_codes}.

\epf

\subsection{Weight Enumerators of Data-Local MRC with $\ell \geq 3$ Local Codes}

\begin{thm}

Consider an $[n,k,d_{\min}]$ data-local MRC with locality $r$, where $k = \ell r$, $n = k + \frac{k}{r} + h = k+\ell + h$, $\ell < r+1$ and $d_{\min} = h+2$. Let the supports of the local groups be $\{ S_1, \ldots, S_\ell \}$. The weight enumerators of the code are given in Table \ref{tab:weight_enum_ell}, where  $N_{u,b}$ = ${\ell \choose b} {n-b(r+1) \choose u-b(r+1)}$, $0 \leq u \leq n-d_{\min} = \ell r + \ell -2$. $N_{u,b} = 0$ if $b(r+1) > u$.

 \begin{table}[ht]
\centering
\begin{tabular}{||c|l||} \hline
   \hline 
   & \\
  $w$ & $\ \ \ \ \ \ \ \ \ \ \ \ \ \ \ \ \ \ \ \ \ \ \ \ \ \  A_w$  \\
  &     \\
  \hline
  & \\
 $ n-\ell r - \ell + 2 \leq w $ & $\sum_{j=0}^{w - h -2} \sum_{b=1}^{w-j-h-1} \left ( N_{n-w+j,b+1} - N_{n-w+j,b+2}   \right ) {n-w+j \choose j} (-1)^j (q^b-1)$ \\
 & \\
   $\leq n- \ell r -1$ & \\
 & \\
  \hline
  & \\
  $n-\ell r \leq w$ &  $\sum_{j=w-n+\ell r + 1}^{w - h -2} \sum_{b=1}^{w-j-h-1} \left ( N_{n-w+j,b+1} - N_{n-w+j,b+2}   \right ) {n-w+j \choose j}  (-1)^j (q^b-1)$ \\
  & \\
  $ \leq n- (\ell-1)(r+1)$ & $+ \sum_{j=0}^{w-n+\ell r}  \sum_{b=0}^{\ell-1} \left ( N_{n-w+j,b} - N_{n-w+j,b+1}   \right ) {n-w+j \choose j} (-1)^j (q^{k-(n-w+j-b)} - 1)$ \\
  & \\
  \hline
  & \\
  $n- (\ell-1)(r+1) + 1 \leq w$ &  $\sum_{j=w-n  +\ell r + 1}^{w - h -2} \sum_{b=1}^{w-j-h-1} \left ( N_{n-w+j,b+1} - N_{n-w+j,b+2}   \right ) {n-w+j \choose j}  (-1)^j (q^b-1)$ \\
  & \\
 $ \leq n- (\ell-2)(r+1)$ & $+ \sum_{j=w-n+(\ell-1)(r+1)}^{w-n+\ell r}  \sum_{b=0}^{\ell-1} \left ( N_{n-w+j,b} - N_{n-w+j,b+1}   \right ) {n-w+j \choose j} (-1)^j (q^{k-(n-w+j-b)} - 1)$ \\
 & \\
 & $+  \sum_{j=0}^{w-n+(\ell-2)(r+1)+r} \sum_{b=0}^{\ell-2}  \left ( N_{n-w+j,b} - N_{n-w+j,b+1}   \right ) {n-w+j \choose j}  (-1)^j (q^{k-(n-w+j-b)} - 1)$ \\
 & \\
  \hline
  & \\
$n- (i+1)(r+1) + 1 \leq w$  &  $\sum_{j=w- n + \ell r + 1}^{w - h -2} \sum_{b=1}^{w-j-h-1} \left ( N_{n-w+j,b+1} - N_{n-w+j,b+2}   \right ) {n-w+j \choose j}(-1)^j (q^b-1)$ \\
& \\
$\leq n- i(r+1), \ \  0 \leq i \leq \ell-3$ & $+ \sum_{j=w-n+(\ell-1)(r+1)}^{w-n+\ell r}  \sum_{b=0}^{\ell-1} \left ( N_{n-w+j,b} - N_{n-w+j,b+1}   \right ) {n-w+j \choose j}  (-1)^j (q^{k-(n-w+j-b)} - 1)$ \\
& \\
& $+ \sum_{j=w-n+(i+1)(r+1)}^{w-n+(i+1)(r+1)+r} \sum_{b=0}^{i+1}  \left ( N_{n-w+j,b} - N_{n-w+j,b+1}   \right ) {n-w+j \choose j}  (-1)^j (q^{k-(n-w+j-b)} - 1)$ \\
& \\
& $+ \sum_{j=0}^{w-n+i(r+1)+r} \sum_{b=0}^{i}  \left ( N_{n-w+j,b} - N_{n-w+j,b+1}   \right ) {n-w+j \choose j} (-1)^j (q^{k-(n-w+j-b)} - 1)$ \\
& \\
  \hline
      \hline
\end{tabular}
\vspace{0.1in}
 \caption{Weight enumerators of data-local MRC with $\ell \geq 3$ local codes.} \label{tab:weight_enum_ell}

\end{table}

\end{thm}

\bpf
As in the case of proof of Theorem \ref{thm:weight_two_codes}, we will evaluate the weight enumerator of the data-local MRC using the following expression:
\begin{equation}
W_C(Z) = 1 + \sum_{\substack{U \subseteq [n]: \\ |U| \leq n-d_{\min}}} Z^{n- |U|} (1-Z)^{|U|} (q^{k-\rho(U)} - 1).
\end{equation}
We will count the number of sets $U \subseteq S$ such that $|U| = u$ and $\rho(U) = v$ for any given pair $(u, v)$ ($0 \leq u \leq n-d_{\min}$ and $0 \leq v \leq k$) based on the matroid structure of the code as given in Proposition \ref{prop:mrc_matroid} . For any set $U$ such that $|U| = u \leq n-d_{\min} = k+ \ell -2 = \ell r + \ell -2$, we have the following cases
\vspace{0.1in}
\ben 
\item $u \leq k$, in which case the maximum rank of the set possible is $u$. For this case, the count of the number of sets with a given cardinality $u$ and rank $v$ is given in Table \ref{tab:ugeqk_datalocal}. The condition $\ell < r+1$ is required to ensure that $(\ell-1)(r+1) \leq \ell r$ in the table.
\item $u > k$, in which case the maximum rank of the set possible is $k$. For this case, the corresponding count is given in Table \ref{tab:ugtk_datalocal}. 
\een

{\small
\begin{table}[ht]
\centering
\begin{tabular}{||c|c|c|c||} \hline
  \hline
    &  &  & \\
  $|U|=u$ & Possible $\rho(U)=v$ & No. of sets $U$ & Remarks \\
  &  &  & \\
  \hline
    &  &   & \\
  $0 \leq u \leq r $ & $u$ & ${n \choose u}$ &  \\
    &  &  &  \\
  \hline
    &  &  &  Total number of sets minus \\
  \multirow{3}{*}{$r+1 \leq u \leq 2r+1$} & $u$ & $ {n \choose u} - {\ell \choose 1} {n-r-1 \choose u-r-1} $&  the number of sets containing \\
      &  &  & at least one local code \\
  \hhline{~---}
      &  &  & Number of sets containing exactly  \\
 &  $u-1$ & ${\ell \choose 1} {n-r-1 \choose u-r-1}$ &  one local code. $r+1$ elements of  \\
        &  &  &   one local code  chosen in $\ell$ ways. \\
        & & & Remaining $u-r-1$ elements \\
        & & &  picked from $n-r-1$ elements \\
\hline
  \vdots & \vdots & \vdots & \\
  \hline
        &  &  & Total number of sets minus  \\

  \multirow{3}{*}{$j(r+1) \leq u \leq (j+1)r+j$} & $u$ & ${n \choose u} - {\ell \choose 1} {n-r-1 \choose u-r-1}$ & the number of sets containing \\
        &  &  & at least one local code \\
  \hhline{~---}
  $2 \leq j \leq \ell -2$    &  &  & Number of sets containing at least   \\
 & $u-1$ & $ {\ell \choose 1} {n-r-1 \choose u-r-1} - {\ell \choose 2} {n-2(r+1) \choose u-2(r+1)}$ & one local code minus those  \\
        &  &  &   containing at least two local codes\\
          \hhline{~---}

& \vdots & \vdots  & \\

  \hhline{~---}
        &  &  & Number of sets containing at least  \\
 & $u-j-1$ & $ {\ell \choose j-1} {n-(j-1)(r+1) \choose u-(j-1)(r+1)} - {\ell \choose j} {n-j(r+1) \choose u-j(r+1)} $ & $j-1$ local codes minus number of  \\

        &  &   & sets containing at least $j$ local codes \\

  \hhline{~---}
        &  & & Number of sets containing exactly  \\
 & $u-j$ & $ {\ell \choose j} {n-j(r+1) \choose u-j(r+1)}$ & $j$ local codes. $j(r+1)$ elements of  \\
&  &  &   $j$ local codes  chosen in ${\ell \choose j}$ ways. \\
        & & & Remaining $u-j(r+1)$ elements \\
        & & &  picked from $n-j(r+1)$ elements \\
  \hline
\vdots & \vdots & \vdots & \\
  \hline
   %     &  &   & \\

  \multirow{3}{*}{$(\ell-1)(r+1) \leq u \leq \ell r$} & $u$ & ${n \choose u} - {\ell \choose 1} {n-r-1 \choose u-r-1}$ & \\
   %     &  &  &  \\
  \hhline{~---}
  %    &  &   & \\
 & $u-1$ & $ {\ell \choose 1} {n-r-1 \choose u-r-1} - {\ell \choose 2} {n-2(r+1) \choose u-2(r+1)}$ & \\
   %     &  &  & \\
          \hhline{~---}

& \vdots & \vdots & \\

  \hhline{~---}
   %     &  &  &  \\
 & $u-(\ell-2)$ & $ {\ell \choose \ell-2} {n-(\ell-2)(r+1) \choose u-(\ell-2)(r+1)} - {\ell \choose \ell-1} {n-(\ell-1)(r+1) \choose u-(\ell-1)(r+1)} $ & \\

    %    &  &  &  \\

  \hhline{~---}
  %      &  &  &  \\
 & $u-(\ell-1)$ & $ {\ell \choose \ell -1} {n-(\ell-1)(r+1) \choose u-(\ell-1)(r+1)}$ & \\

   %     &  &  & \\
  \hline
      \hline
\end{tabular}
\vspace{0.1in}
 \caption{Count of the number of sets $U \subseteq S$ such that $|U| = u$ and $\rho(U) = v$ for $0 \leq u \leq k=\ell r$.} \label{tab:ugeqk_datalocal}

\end{table}

}

\begin{table}[ht]
\centering
\begin{tabular}{||c|c|c|c||} \hline
  \hline
    &  &  &  \\
  $|U|=u$ & Possible $\rho(U)=v$ & No. of sets $U$ & Remarks \\
  &  &  &  \\
  \hline
    &  &   & Rank drops by $1$ when \\
   \multirow{3}{*}{$\ell r +1$} & $k-1$ & $ {\ell \choose 2} {n-2(r+1) \choose u-2(r+1)} -  {\ell \choose 3} {n-3(r+1) \choose u-3(r+1)}$ & exactly two local codes are  \\
    &  &  & included in the set \\
  \hhline{~---}
     & \vdots & \vdots & \\

    \hhline{~---}
        &  &  &  Rank drops by $\ell-2$ when \\

    & $k-(\ell-2)$ & $ {\ell \choose (\ell-1)} {n-(\ell-1)(r+1) \choose u-(\ell-1)(r+1)}$ & exactly $\ell-1$ local codes are \\
        &  & & included in the set  \\

  \hline
    \vdots & \vdots & \vdots &  \\
    \hline
        &  &  & Rank drops by $1$ when \\

       \multirow{3}{*}{$\ell r +j$} & $k-1$ & $ {\ell \choose j+1} {n-(j+1)(r+1) \choose u-(j+1)(r+1)} -  {\ell \choose j+2} {n-(j+2)(r+1) \choose u-(j+2)(r+1)}$ & exactly $j+1$ local codes are\\
    &  &  & included in the set \\
  \hhline{~---}
     & \vdots & \vdots & \\

    \hhline{~---}
        &  &  & Rank drops by $\ell-1-j$ when \\

    & $k-(\ell-1-j)$ & $ {\ell \choose \ell-1} {n-(\ell-1)(r+1) \choose u-(\ell-1)(r+1)}$  & exactly $\ell - 1$ local codes are\\
        &  &   & included in the set \\

    \hline

        \vdots & \vdots & \vdots & \\

  \hline

    &  &   & Rank drops by $1$ when \\
  $\ell r + \ell - 2$ & $k-1$ & $ {\ell \choose \ell - 1} {n-(\ell-1)(r+1) \choose u-(\ell-1)(r+1)}$ & exactly $\ell-1$ local codes are \\
      &  &  &  included in the set \\
  
\hline
      \hline
\end{tabular}
\vspace{0.1in}
 \caption{Count of the number of sets $U \subseteq S$ such that $|U| = u$ and $\rho(U) = v$ for $k+1 \leq u \leq k+ \ell -2$. We have not included the count corresponding to rank $k$ as that term reduces to zero.} \label{tab:ugtk_datalocal}

\end{table}

Let $N_{u,b}$ = ${\ell \choose b} {n-b(r+1) \choose u-b(r+1)}$, $0 \leq u \leq n-d_{\min} = \ell r + \ell -2$. $N_{u,b} = 0$ if $b(r+1) > u$.
The weight enumerator polynomial of data-local MRC with $\ell$ local codes can be rewritten as
{\small
\begin{eqnarray*}
W_C(Z) &= &1+ \sum_{i=0}^{\ell-2} \sum_{u=i(r+1)}^{i(r+1)+r} \sum_{b=0}^{i}  \left ( N_{u,b} - N_{u,b+1}   \right ) Z^{n-u} (1-Z)^u (q^{k-(u-b)} - 1) \\
&& + \sum_{u=(\ell-1)(r+1)}^{\ell r}  \sum_{b=0}^{\ell-1} \left ( N_{u,b} - N_{u,b+1}   \right ) Z^{n-u} (1-Z)^u (q^{k-(u-b)} - 1) \\
& & + \sum_{u=\ell r+1}^{\ell r+\ell-2} \sum_{b=1}^{\ell-1-u+k} \left ( N_{u,b+1} - N_{u,b+2}   \right ) Z^{n-u} (1-Z)^u (q^b-1).
\end{eqnarray*}
}
{\small
\begin{eqnarray*}
W_C(Z) & = & 1+ \sum_{i=0}^{\ell-2} \sum_{u=i(r+1)}^{i(r+1)+r} \sum_{j=0}^u \sum_{b=0}^{i}  \left ( N_{u,b} - N_{u,b+1}   \right ) {u \choose j} Z^{n-u+j} (-1)^j (q^{k-(u-b)} - 1) \\
&& + \sum_{u=(\ell-1)(r+1)}^{\ell r} \sum_{j=0}^u  \sum_{b=0}^{\ell-1} \left ( N_{u,b} - N_{u,b+1}   \right ) {u \choose j} Z^{n-u+j} (-1)^j (q^{k-(u-b)} - 1) \\
& & + \sum_{u=\ell r+1}^{\ell r+\ell-2} \sum_{j=0}^u \sum_{b=1}^{\ell-1-u+k} \left ( N_{u,b+1} - N_{u,b+2}   \right ) {u \choose j} Z^{n-u+j} (-1)^j (q^b-1) \\
& = & 1 + \sum_{i=0}^{\ell-2} T_{1i} + T_2 + T_3.
\end{eqnarray*}
}

 By applying the change of variables $w = n-u+j, j=j$ to the above equation, we first identify the $(u,j)$ pairs which result in a given value of $w$ (Table \ref{tab:change_variable_three}).

 \footnotesize{
 \begin{eqnarray*}
W_C(Z) & = & 1+ \sum_{i=0}^{\ell-2} \sum_{w=n-i(r+1)-r}^{n-i(r+1)} \sum_{j=0}^{w-n+i(r+1)+r} \sum_{b=0}^{i}  \left ( N_{n-w+j,b} - N_{n-w+j,b+1}   \right ) {n-w+j \choose j} Z^{w} (-1)^j (q^{k-(n-w+j-b)} - 1) \\
&& + \sum_{i=0}^{\ell-2} \sum_{w=n-i(r+1)+1}^{n} \sum_{j=w-n+i(r+1)}^{w-n+i(r+1)+r} \sum_{b=0}^{i}  \left ( N_{n-w+j,b} - N_{n-w+j,b+1}   \right ) {n-w+j \choose j} Z^{w} (-1)^j (q^{k-(n-w+j-b)} - 1) \\
&& + \sum_{w=n-\ell r}^{n-(\ell-1)(r+1)} \sum_{j=0}^{w-n+\ell r}  \sum_{b=0}^{\ell-1} \left ( N_{n-w+j,b} - N_{n-w+j,b+1}   \right ) {n-w+j \choose j} Z^{w} (-1)^j (q^{k-(n-w+j-b)} - 1) \\
& & + \sum_{w=n-(\ell-1)(r+1)+1}^{n} \sum_{j=w-n+(\ell-1)(r+1)}^{w-n+\ell r}  \sum_{b=0}^{\ell-1} \left ( N_{n-w+j,b} - N_{n-w+j,b+1}   \right ) {n-w+j \choose j} Z^{w} (-1)^j (q^{k-(n-w+j-b)} - 1) \\
& & + \sum_{w=n-\ell r-\ell+2}^{n-\ell r -1} \sum_{j=0}^{w - n +\ell r + \ell -2} \sum_{b=1}^{\ell-1-n+w-j+k} \left ( N_{n-w+j,b+1} - N_{n-w+j,b+2}   \right ) {n-w+j \choose j} Z^{w} (-1)^j (q^b-1) \\
& & + \sum_{w=n-\ell r}^{n} \sum_{j=w-n+\ell r + 1}^{w - n +\ell r + \ell -2} \sum_{b=1}^{\ell-1-n+w-j+k} \left ( N_{n-w+j,b+1} - N_{n-w+j,b+2}   \right ) {n-w+j \choose j} Z^{w} (-1)^j (q^b-1).
\end{eqnarray*}

 }
 
 \normalsize
 The weight enumerators in the theorem statement follow by bringing together coefficients with the same power of $Z$ and by noting that $n-\ell r - \ell + 2 = h+2$ and $\ell - 1-n+w-j+k = w-j-h-1$.

  \begin{table}[h]
\centering
\begin{tabular}{||c|c|c|l|c||} \hline
   \hline
  Term & $w$ & No. of  & $(u,j)$ pairs & Range of $j$ \\
  in Sum & & terms & & \\
  \hline
  \multirow{9}{*}{$T_{1i}$} & $n$ & $r+1$ & $(i(r+1),i(r+1)), \ldots, $ &  \\
  & & & $\ \ \ \ \ \ \ \ \ \ \ (i(r+1)+r,i(r+1)+r)$ &  $j=w-n+i(r+1)$ \\
  \hhline{~---~}
  & $n-1$ & $r+1$ & $(i(r+1),i(r+1)-1),  \ldots, $ &  \\
    & & & $\ \ \ \ \ \ \ \ \ \ \ (i(r+1)+r, i(r+1)+r-1)$ & to $j=w-n+i(r+1)+r$ \\
\hhline{~---~}
  & \vdots & \vdots & \vdots &  \\
  \hhline{~---~}
  & $n-i(r+1)+1$ & $r+1$ &  $(i(r+1),1),  \ldots, (i(r+1)+r,r+1)$ &  \\
  \hhline{~----}
& $n-i(r+1)$ & $r+1$ &  $(i(r+1),0),  \ldots, (i(r+1)+r,r)$ & $j=0$ to \\
  \hhline{~---~}
 & \vdots & \vdots & \vdots & $j=w-n+i(r+1)+r$ \\
  \hhline{~---~}
& $n-i(r+1)-r$ & $1$ & $(i(r+1)+r,0)$ & \\
        \hline
      \hline
      
 \multirow{7}{*}{$T_2$} & $n$ & $r+2-\ell$ & $((\ell-1)(r+1),(\ell-1)(r+1)), \ldots, (\ell r, \ell r)$ & \\
  \hhline{~---~}
  & $n-1$ & $r+2-\ell$ & $((\ell-1)(r+1),(\ell-1)(r+1)-1),  \ldots, (\ell r, \ell r - 1)$ & $j=w-n+(\ell-1)(r+1)$ \\
  \hhline{~---~}
 & \vdots & \vdots & \vdots & to $j=w-n+\ell r$ \\
  \hhline{~---~}
  & $n-(\ell-1)(r+1)+1$ & $r+2-\ell$ &  $((\ell-1)(r+1),1),  \ldots, (\ell r,r+2-\ell)$ & \\
  \hhline{~----}
& $n-(\ell-1)(r+1)$ & $r+2-\ell$ &  $((\ell-1)(r+1),0),  \ldots, (\ell r,r+1-\ell)$ & $j=0$ to \\
  \hhline{~---~}
  & \vdots & \vdots & \vdots &  $j=w-n+\ell r $ \\
  \hhline{~---~}
& $n-\ell r$ & $1$ & $(\ell r,0)$ & \\
        \hline
      \hline
   \multirow{7}{*}{$T_3$} & $n$ & $\ell-2$ & $(\ell r +1,\ell r +1), \ldots, (\ell r + \ell - 2, \ell r + \ell - 2)$ & \\
  \hhline{~---~}
  & $n-1$ & $\ell-2$ & $(\ell r +1,\ell r),  \ldots, (\ell r + \ell - 2, \ell r + \ell - 3)$ & $j=w-n+\ell r +1$ \\
  \hhline{~---~}
 & \vdots & \vdots & \vdots & to $j=w-h-2$ \\
  \hhline{~---~}
  & $h  +\ell$ & $\ell-2$ &  $(\ell r+1,1),  \ldots, (\ell r + \ell - 2, \ell - 2)$ & \\
  \hhline{~----}
& $h  +\ell - 1$ & $\ell-2$ &  $(\ell r+1,0),  \ldots, (\ell r + \ell - 2, \ell - 3)$ & $j=0$ to \\
  \hhline{~---~}
  & \vdots & \vdots & \vdots &  $j=w-h-2 $ \\
  \hhline{~---~}
& $h+2$ & $1$ & $(\ell r + \ell - 2,0)$ & \\
        \hline
      \hline

\end{tabular}
\vspace{0.1in}
 \caption{Illustrating the change of variables from $u,j$ to $w=n-u+j, j$ for a data-local MRC with $\ell \geq 3$ local codes.} \label{tab:change_variable_three}

\end{table}

% \begin{table}[ht]
%\centering
%\begin{tabular}{||c|c|l||} \hline
%   \hline
%  $w$ & No. of terms & $(u,j)$ pairs \\
%  &  &   \\
%  \hline
%  $n$ & $r+2-\ell$ & $((\ell-1)(r+1),(\ell-1)(r+1)), \ldots, (\ell r, \ell r)$ \\
%  \hline
%  $n-1$ & $r+2-\ell$ & $((\ell-1)(r+1),(\ell-1)(r+1)-1),  \ldots, (\ell r, \ell r - 1)$ \\
%  \hline
%  \vdots & \vdots & \vdots \\
%  \hline
%  $n-(\ell-1)(r+1)+1$ & $r+2-\ell$ &  $((\ell-1)(r+1),1),  \ldots, (\ell r,r+2-\ell)$ \\
%  \hline
%$n-(\ell-1)(r+1)$ & $r+2-\ell$ &  $((\ell-1)(r+1),0),  \ldots, (\ell r,r+1-\ell)$ \\
%  \hline
%  \vdots & \vdots & \vdots  \\
%  \hline
%$n-\ell r$ & $1$ & $(\ell r,0)$ \\
%        \hline
%      \hline
%\end{tabular}
%\vspace{0.1in}
% \caption{Applying the change of variable $w = n-u+j, j=j$.  Note that in the first term, the range of $j$ for a given $w$ is $0$ to $w-h-3$. 
% In the second term, for $h+2 \leq w \leq r+h+1$, the range of $j$ given $w$ is $0$ to $w-h-2$.  For $r+h+2 \leq w \leq n$, the range of $j$ given $w$ is $w-r-h-1$ to $w-h-2$.} \label{tab:change_variable_two}
%
%\end{table}

\epf

\section{Weight Enumerators of Local MRC} \label{sec:weights_local_mrc}

In this section, we identify the matroid structure of a local MRC. We derive the weight enumerators of a local MRC with two local codes.

\begin{prop} \label{prop:local_mrc_matroid}
Consider an $[n,k,d_{\min}]$ local MRC with locality $r$, where $n = k + h + \frac{k+h}{r} $ and $d_{\min} = h+2$. Let the supports of the $\frac{k+h}{r}$ local groups be $\{ S_1, S_2, \ldots, S_{\frac{k+h}{r}} \}$. The set of all independent sets of the local MRC are given by
\begin{equation}
\mathcal{I} = \{ I \subseteq S \ | \ |I| \leq k, |I \cap S_i| \leq r, 1 \leq i \leq \frac{k+h}{r}\}.
\end{equation}
\end{prop}
\bpf
The proof is similar to that of Proposition \ref{prop:mrc_matroid}.
\epf

\subsection{Weight Enumerators of Local MRC with Two Local Codes}

\begin{thm} \label{thm:allsymbol_two_codes}
Consider an $[n,k,d_{\min}]$ local MRC with locality $r$, where $k+h = 2r$, $r+1 \leq k < 2r$, $n = k +h+ \frac{k+h}{r}$ and $d_{\min} = h+2$. Let the supports of the two local groups be $\{ S_1, S_2 \}$. The weight enumerators of the code are given by
\begin{eqnarray*}
A_{h+2} & =  & 2 {n-r-1 \choose r-1}(q-1), \\
A_w & = &\sum_{j=0}^{w-h-3} {n \choose w}{w \choose j} (-1)^j (q^{k+w-n-j}-1) \\
&&+ \sum_{j=0}^{w-h-2} 2 (q-1) {n-r-1 \choose n-w+j-r-1} {n-w+j \choose j} (-1)^j q^{k+w-n-j}, \ \ \ \ \ \ h+3 \leq w \leq r+1, \\
A_w & = &\sum_{j=0}^{w-h-3} {n \choose w}{w \choose j} (-1)^j (q^{k+w-n-j}-1) \\
&&+ \sum_{j=w-r-1}^{w-h-2} 2 (q-1) {n-r-1 \choose n-w+j-r-1} {n-w+j \choose j} (-1)^j q^{k+w-n-j}, \ \ r+2 \leq w \leq n.
\end{eqnarray*}

\end{thm} 

\begin{proof}
The proof of this theorem is similar to that of Theorem \ref{thm:weight_two_codes}. See Appendix \ref{app:allsymbol_two_codes}.

\end{proof}

\begin{note}
The weight enumerator of a local MRC with $\ell \geq 3$ local codes can be calculated in a similar way as in the case of data-local MRC with $\ell \geq 3$ local codes. The derivation and expressions of weight enumerators are avoided for lack of space.
\end{note}

\section{Higher Support Weights of MRC} \label{sec:hsw_mrc}

In this section, we define generalized Hamming weights (GHW) and present two properties of GHWs. Then, we  present the result by Britz \cite{Bri_hsw}, which relates the Tutte polynomial of a code to all the higher support weights. 
Subsequently, we give the GHW structure of a data-local MRC with two local codes and use them to derive the higher support weights of the code. Finally, we use the same techniques to derive the higher support weights of a local MRC with two local codes.

\subsection{Generalized Hamming Weights, Higher Support Weights and Tutte Polynomial}

\begin{defn}[Generalized Hamming Weights]
The $i^{th}$, $1 \leq i \leq k$, generalized Hamming weight~\cite{Wei} (also known as minimum support weight~\cite{HelKloLevYtr}) of a code ${\cal C}$ is defined by
\begin{equation}
d_i({\cal C}) \ = \ d_i \ = \  \min_{\substack{ \mathcal{D} < \mathcal{C} \\ \text{dim}(\mathcal{D}) = i }}
\left|\text{Supp}({\cal D}) \right| ,
\end{equation}
where $\mathcal{D} < \mathcal{C}$, is used to denote a subcode $\mathcal{D}$ of $\mathcal{C}$.
\end{defn}

\vspace{0.1in}

It is well known that $d=d_1 < d_2 < \ldots < d_k = n$.

\vspace{0.1in}

\begin{lem} \label{lem:GHW_code_dual_relation}
The GHWs of $\mathcal{C}$ are related to those of $\mathcal{C}^{\perp}$ by the following relation \cite{Wei}:
 \begin{equation} \label{eq:GHW_code_dual_relation}
 \{d_i, \ 1 \leq i \leq k\} \ = \ [n] \setminus \{n + 1 - d_j^{\perp}, \ \ 1 \leq j \leq n-k \}.
 \end{equation}
\end{lem}

\begin{lem} \label{lem:ghw_cardinality}
For any set $U \subseteq [n]$ such that $|U| \geq n-d_{i} + 1, 1 \leq i \leq k$, $\rho(U) \geq k-(i-1)$.
\end{lem}
\begin{proof}
Consider the generator matrix including all basis vectors of sub code shortened on $[n] \setminus U$. The matrix can be written as follows
\begin{equation}
G = \left [ \begin{array}{c|c} A & 0 \\ \hline \underbrace{B}_{d_i - 1 \ \text{columns}} & \underbrace{C}_{U} \end{array} \right ]
\end{equation}
We have that $\text{rank}(A) \leq i-1$ and $\text{rank}(A) + \text{rank}(C) = k$. Hence, we have $\text{rank}(C) \geq k-(i-1)$.
\end{proof}

For $1 \leq s \leq k$, the $s^{th}$ support weight polynomial of a code $\mathcal{C}$ is defined as a polynomial $W^{(s)}_C(Z) = \sum_{i=0}^n A_i^{(s)} Z^i$,
where $A_i^{(s)} = | D \in \mathcal{C} : \text{dim}(D) = s, \ wt(D) = i|$. For integers $s,t \geq 0$, we define
\begin{equation}
[t]_s = \prod_{i=0}^{s-1} (q^t - q^i) \ \ \  \text{and} \ \ \ {t \brack s} = \frac{[t]_s}{[s]_s},
\end{equation}
${s \brack t}$ indicate the number of $t$ dimensional subspaces of a given $s$ dimensional subspace over $\mathbb{F}_q$. ${s \brack t} = 0$ when $s < t$.

The higher support weights of a code are related to the Tutte polynomial of the code as follows (Britz's result \cite{Bri_hsw}, \cite{Bri_critical}, \cite{Bri_tutte}):

\begin{thm}
For $1 \leq s \leq k$, the $s^{th}$ support weight polynomial of a code can be determined from the Tutte polynomial of the code $T_C(X,Y)$ using the following relation:
\begin{equation}
W^{(s)}_C(Z) = Z^{n-k}(1-Z)^k \sum_{i=0}^s \frac{(-1)^{s-i}}{[s]_s} q^{{s-i \choose 2}} {s \brack i} T_C( \frac{1+(q^i-1)Z}{1-Z}, \frac{1}{Z} ),
\end{equation}
\end{thm}
The above equation can be simplified as follows.

\begin{eqnarray}
W^{(s)}_C(Z) & = & Z^{n-k}(1-Z)^k \sum_{i=0}^s \frac{(-1)^{s-i}}{[s]_s} q^{{s-i \choose 2}} {s \brack i} T_C( \frac{1+(q^i-1)Z}{1-Z}, \frac{1}{Z} ) \nonumber \\
& = &      Z^{n-k}(1-Z)^k   \sum_{i=0}^s \frac{(-1)^{s-i}}{[s]_s} q^{{s-i \choose 2}} {s \brack i}   \sum_{U \subseteq [n]} \left ( \frac{1+(q^i-1)Z}{1-Z} -1 \right )^{k - \rho(U)}\left (\frac{1}{Z}-1 \right )^{|U| - \rho(U)} \nonumber \\
& = &  Z^{n-k}(1-Z)^k    \sum_{i=0}^s \frac{(-1)^{s-i}}{[s]_s} q^{{s-i \choose 2}} {s \brack i}  \sum_{U \subseteq [n]} \left ( \frac{q^i Z}{1-Z} \right )^ {k - \rho(U)}  \left (\frac{1}{Z}-1 \right )^{|U| - \rho(U)} \nonumber \\
& = &  \sum_{U \subseteq [n]} Z^{n- |U|} (1-Z)^{|U|} \sum_{i=0}^s \frac{(-1)^{s-i}}{[s]_s} q^{{s-i \choose 2}} {s \brack i} q^{i(k-\rho(U))}. \\
& \stackrel{(a)}{=} & \sum_{U \subseteq [n]} Z^{n- |U|} (1-Z)^{|U|} {k-\rho(U) \brack s}, 
\end{eqnarray}

where $(a)$ follows from the following identity \cite{Bri_critical},
\begin{equation}
{t \brack s}  =   \sum_{i=0}^s \frac{(-1)^{s-i}}{[s]_s} q^{{s-i \choose 2}} {s \brack i} q^{it}
\end{equation}

For any set $U$ such that $|U| \geq n-d_{s} + 1$, $\rho(U) \geq k-(s-1)$  and hence ${k-\rho(U) \brack s} = 0$. Hence, the $s^{th}$ support weight polynomial  can be rewritten as
\begin{equation} \label{eq:higher_support}
W^{(s)}_C(Z) = \sum_{\substack{U \subseteq [n]: \\ |U| \leq n-d_{s}}} Z^{n- |U|} (1-Z)^{|U|} {k-\rho(U) \brack s}.
\end{equation}

\subsection{Higher Support Weights of Data-Local MRC with Two Local Codes}

\begin{prop} \label{prop:ghw_matroid}
Consider an $[n,k,d_{\min}]$ data-local MRC $\mathcal{C}$ with locality $r$, where $k = 2r$, $n = k + \frac{k}{r} + h$ and $d_{\min} = h+2$. The generalized Hamming weights of the code  are given by
\begin{equation}
d_s = \begin{cases} h+1+s, & 1 \leq s \leq r, \\
h + 2+s, & r+1 \leq s \leq 2r. \end{cases}
\end{equation}
\end{prop}
\begin{proof}
Consider the sub code $\mathcal{D}$ of the dual code  $\mathcal{C}^\perp$ spanned by the two local parities. It can be verified that $\mathcal{C}$ is a data-local MRC if and only if any $k$-core of $\mathcal{D}$ is also $k$-core of $\mathcal{C}^\perp$.
Applying Theorem 4.3 of \cite{PraLalKum_arxiv}, we have that generalized Hamming weights of $\mathcal{C}^\perp$ are given by
\begin{eqnarray}
d_1^{\perp} & = & r+1 \nonumber \\
d_i^{\perp} & = & 2r+i, \ \ 2 \leq i \leq h+2.
\end{eqnarray}
The theorem follows by applying the duality of generalized Hamming weights from Lemma \ref{lem:GHW_code_dual_relation}.
\end{proof}

\begin{thm} \label{thm:hsw_two_codes}
Consider an $[n,k,d_{\min}]$ data-local MRC with locality $r$, where $k = 2r$, $n = k + \frac{k}{r} + h$ and $d_{\min} = h+2$. For $1 \leq s \leq r$, the higher support weights of the code are given in Table \ref{tab:hsw_data_two}.
%\begin{eqnarray*}
%A^{(s)}_{h+1+s} & =  & 2 {n-r-1 \choose r-1}(q-1), \\
%A_w & = &\sum_{j=0}^{w-h-2-s} {n \choose w}{w \choose j} (-1)^j {k+w-n-j \brack s} \\
%&&+ \sum_{j=0}^{w-h-1-s} 2  {n-r-1 \choose n-w+j-r-1} {n-w+j \choose j} (-1)^j ({k+w-n-j+1 \brack s}-{k+w-n-j \brack s}), \ \ \ \ \ \ h+2+s \leq w \leq r+h+1, \\
%A^{(s)}_w & = &\sum_{j=0}^{w-h-2-s} {n \choose w}{w \choose j} (-1)^j {k+w-n-j \brack s} \\
%&&+ \sum_{j=w-r-h-1}^{w-h-1-s} 2  {n-r-1 \choose n-w+j-r-1} {n-w+j \choose j} (-1)^j ({k+w-n-j+1 \brack s}-{k+w-n-j \brack s}), \ \ r+h+2 \leq w \leq n.
%\end{eqnarray*}
%
%
%
 \begin{table}[ht]
\centering
\begin{tabular}{||c|l||}
   \hline 
   \multicolumn{2}{||c||}{Higher support weights for $2 \leq s \leq r-1$} \\
   \hline
  $w$ & $\ \ \ \ \ \ \ \ \ \ \ \ \ \ \ \ \ \ \ \ \ \ \ \ \ \  A^{(s)}_w$  \\
  &     \\
  \hline
  & \\
  $h+1+s$ &  $2  {n-r-1 \choose r-s}$ \\
  & \\
  \hline
  & \\
  $h+2+s \leq w$ &  $\sum_{j=0}^{w-h-2-s} {n \choose w}{w \choose j} (-1)^j {k+w-n-j \brack s}$ \\
  & \\
 $ \leq r+h+1$ & $+ \sum_{j=0}^{w-h-1-s} 2  {n-r-1 \choose n-w+j-r-1} {n-w+j \choose j} (-1)^j ({k+w-n-j+1 \brack s}-{k+w-n-j \brack s})$ \\
 & \\
  \hline
  & \\
$r+h+2 \leq w$  &  $\sum_{j=0}^{w-h-2-s} {n \choose w}{w \choose j} (-1)^j {k+w-n-j \brack s}$ \\
& \\
$\leq n$ & $+ \sum_{j=w-r-h-1}^{w-h-1-s} 2  {n-r-1 \choose n-w+j-r-1} {n-w+j \choose j} (-1)^j ({k+w-n-j+1 \brack s}-{k+w-n-j \brack s})$ \\
& \\
  \hline
   \multicolumn{2}{||c||}{Higher support weights for $s = r$} \\
     \hline
  $w$ & $\ \ \ \ \ \ \ \ \ \ \ \ \ \ \ \ \ \ \ \ \ \ \ \ \ \  A^{(r)}_w$  \\
  &     \\
  \hline

  & \\
  $h+1+r$ &  $2$ \\
  & \\
  \hline
  & \\
  $r+h+2 \leq w$ &  $\sum_{j=0}^{w-h-2-s} {n \choose w}{w \choose j} (-1)^j {k+w-n-j \brack s}$ \\
  & \\
 $ \leq n$ & $+ \sum_{j=w-r-h-1}^{w-h-1-s} 2  {n-r-1 \choose n-w+j-r-1} {n-w+j \choose j} (-1)^j ({k+w-n-j+1 \brack s}-{k+w-n-j \brack s})$ \\
 & \\
  \hline
\end{tabular}
\vspace{0.1in}
 \caption{Higher support weights of data-local MRC with two local codes.} \label{tab:hsw_data_two}

\end{table}

\end{thm}

\begin{proof}

We have to consider all sets $U$ such that $|U| = u \leq n-d_{s} = 2r + h + 2 - (h+1+s) = 2r-s+1$ (applying Proposition \ref{prop:ghw_matroid}). Note that  $2r-s+1 \geq r+1$, since  $2 \leq s \leq r$. The proof follows from the same counting as in the proof of Theorem \ref{thm:weight_two_codes}.
\vspace{0.1in}
\ben
\item For a set $U$ of size $0 \leq u \leq r$, $\rho(U) = u$. The number of such sets are ${n \choose u}$. 
\vspace{0.05in}
\item For a set $U$ of size $r+1 \leq u \leq 2r-s+1$, there are the following two cases possible.   
\ben
\item The number of sets with $\rho(U) = u-1$ are $2{n-r-1 \choose u-r-1}$.
\vspace{0.05in}
\item The number of sets with $\rho(U) = u$ are ${n \choose u} - 2{n-r-1 \choose u-r-1}$. 
\een
\een
Hence, the $s^{th}$ support weight polynomial in \eqref{eq:higher_support} can be rewritten as
\begin{eqnarray}
W^{(s)}_C(Z) &= & \sum_{u=0}^r {n \choose u} Z^{n-u} (1-Z)^u {k-u \brack s} \nonumber \\
&& + \sum_{u=r+1}^{2r-s+1} 2{n-r-1 \choose u-r-1} Z^{n-u} (1-Z)^u {k-u+1 \brack s} \nonumber \\
&& + \sum_{u=r+1}^{2r-s+1} \left ( {n \choose u} - 2{n-r-1 \choose u-r-1} \right ) Z^{n-u} (1-Z)^u {k-u \brack s} \nonumber \\
& \stackrel{(a)}{=} & \sum_{u=0}^{2r-s} \sum_{j=0}^{u} {n \choose u} {u \choose j} Z^{n-u+j} (-1)^j {k-u \brack s} \nonumber \\
&& + \sum_{u=r+1}^{2r-s+1} \sum_{j=0}^{u} 2{n-r-1 \choose u-r-1} {u \choose j} Z^{n-u+j} (-1)^j ({k-u+1 \brack s} - {k-u \brack s}),
\end{eqnarray}
where $(a)$ follows since for the case when $u=2r-s+1$ and the set is full rank, the term in the summation is zero.
 By applying the change of variables $w = n-u+j, j=j$ to the above equation, the higher support weights given in the theorem statement follow.

 \begin{note}
 The relation between $W^{(1)}(Z)$ and the weight enumerator polynomial of a code $W_C(Z)$ is given by
 \begin{equation}
 W_C(Z) = 1 + (q-1) W_C^{(1)}(Z).
 \end{equation}
 
 \end{note}

% \begin{table}[ht]
%\centering
%\begin{tabular}{||c|c|c|l|c||}
%\hline
%\hline
%  Term & $w$ & No. of  & $(u,j)$ pairs & Range of $j$ \\
%  in Sum & & terms & & \\
%  \hline
%        
% \multirow{4}{*}{$T_1$} & $n$ & $2r-s$ & $(0,0), (1,1), \ldots, (2r-s,2r-s)$ & \\
%  \hhline{~---~}
%  & $n-1$ & $2r-s-1$ & $(1,0), (2,1), \ldots, (2r-s, 2r-s-1)$ & $j=0$ to\\
%  \hhline{~---~}
% & \vdots & \vdots & \vdots & $j=w-h-2-s$ \\
%  \hhline{~---~}
%& $h+2+s$ & $1$ & $(2r-s,0)$ & \\
%\hline
%\hline
%  \multirow{7}{*}{$T_2$} & $n$ & $r-s+1$ & $(r+1,r+1), (r+2,r+2), \ldots, (2r-s+1,2r-s+1)$  &  \\
%  \hhline{~---~}
%  & $n-1$ & $r-s+1$ & $(r+1,r), (r+2,r+1), \ldots, (2r-s+1, 2r-s)$ &  $j=w-r-h-1$ \\
%  \hhline{~---~}
%  & \vdots & \vdots & \vdots & to $j=w-h-1-s$ \\
%  \hhline{~---~}
%  & $r+h+2$ & $r-s+1$ &  $(r+1,1), (r+2,r-1), \ldots, (2r-s+1,r-s+1)$ &  \\
%  \hhline{~----}
%& $r+h+1$ & $r-s+1$ &  $(r+1,0), (r+2,1), \ldots, (2r-s+1,r-s)$ & $j=0$ to \\
%  \hhline{~---~}
% & \vdots & \vdots & \vdots & $j=w-h-1-s$ \\
%  \hhline{~---~}
%& $h+1+s$ & $1$ & $(2r-s+1,0)$ & \\
%        \hline
%      \hline
%
%\end{tabular}
%
%\end{table}

\end{proof}

\begin{thm} \label{thm:hsw_two_codes}
Consider an $[n,k,d_{\min}]$ data-local MRC with locality $r$, where $k = 2r$, $n = k + \frac{k}{r} + h$ and $d_{\min} = h+2$. For $r+1 \leq s \leq 2r$, the higher support weights of the code are given by
\begin{eqnarray}
A^{(s)}_w = {n \choose w} \sum_{j=0}^{w-n+2r-s}  {w \choose j} (-1)^j {k-n+w-j \brack s}_q, \ \  h+2+s \leq w \leq n.
\end{eqnarray}

\end{thm}

\begin{proof} We have to consider all sets $U$ such that $|U| = u \leq n-d_{s} = 2r + h + 2 - (h+2+s) = 2r-s$ (applying Proposition \ref{prop:ghw_matroid}). Note that  $2r-s \leq r-1$, since  $r+1 \leq s \leq 2r$. 
The $s^{th}$ support weight polynomial in \eqref{eq:higher_support} can be rewritten as
\begin{eqnarray}
W^{(s)}_C(Z) &= & \sum_{u=0}^{2r-s} {n \choose u} Z^{n-u} (1-Z)^u {k-u \brack s} \nonumber \\
& = & \sum_{u=0}^{2r-s} \sum_{j=0}^{u} {n \choose u} {u \choose j} Z^{n-u+j} (-1)^j {k-u \brack s} \nonumber \\
& = & \sum_{w=n-2r+s}^{n} \sum_{j=0}^{w-n+2r-s} {n \choose w} {w \choose j} Z^{w} (-1)^j {k-n+w-j \brack s}.
\end{eqnarray}

\end{proof}

\subsection{Higher Support Weights of Local MRC with Two Local Codes}

\begin{prop} \label{prop:ghw_matroid_local}
Consider an $[n,k,d_{\min}]$ local MRC with locality $r$, where $k+h = 2r$, $r+1 \leq k < 2r$, $n = k +h+ \frac{k+h}{r}$ and $d_{\min} = h+2$. The generalized Hamming weights of the code \cite{PraLalKum_arxiv} are given by
\begin{equation}
d_s = \begin{cases} h+1+s, & 1 \leq s \leq r-h, \\
h + 2+s, & r-h+1 \leq s \leq 2r-h. \end{cases}
\end{equation}
\end{prop}
\begin{proof}
The proof is similar to the proof of Proposition \ref{prop:ghw_matroid}.
\end{proof}

\begin{thm} \label{thm:hsw_two_codes_local}
Consider an $[n,k,d_{\min}]$ local MRC with locality $r$, where $k+h = 2r$, $r+1 \leq k < 2r$, $n = k +h+ \frac{k+h}{r}$ and $d_{\min} = h+2$. For $1 \leq s \leq r-h$, the higher support weights of the code are given in Table \ref{tab:hsw_data_two_local}.

 \begin{table}[ht]
\centering
\begin{tabular}{||c|l||}
   \hline 
   \multicolumn{2}{||c||}{Higher support weights for $1 \leq s \leq r-h-1$} \\
   \hline
  $w$ & $\ \ \ \ \ \ \ \ \ \ \ \ \ \ \ \ \ \ \ \ \ \ \ \ \ \  A^{(s)}_w$  \\
  &     \\
  \hline
  & \\
  $h+1+s$ &  $2  {n-r-1 \choose r-s}$ \\
  & \\
  \hline
  & \\
  $h+2+s \leq w$ &  $\sum_{j=0}^{w-h-2-s} {n \choose w}{w \choose j} (-1)^j {k+w-n-j \brack s}$ \\
  & \\
 $ \leq r+1$ & $+ \sum_{j=0}^{w-h-1-s} 2  {n-r-1 \choose n-w+j-r-1} {n-w+j \choose j} (-1)^j ({k+w-n-j+1 \brack s}-{k+w-n-j \brack s})$ \\
 & \\
  \hline
  & \\
$r+2 \leq w$  &  $\sum_{j=0}^{w-h-2-s} {n \choose w}{w \choose j} (-1)^j {k+w-n-j \brack s}$ \\
& \\
$\leq n$ & $+ \sum_{j=w-r-1}^{w-h-1-s} 2  {n-r-1 \choose n-w+j-r-1} {n-w+j \choose j} (-1)^j ({k+w-n-j+1 \brack s}-{k+w-n-j \brack s})$ \\
& \\
  \hline
   \multicolumn{2}{||c||}{Higher support weights for $s = r-h$} \\
     \hline
  $w$ & $\ \ \ \ \ \ \ \ \ \ \ \ \ \ \ \ \ \ \ \ \ \ \ \ \ \  A^{(r-h)}_w$  \\
  &     \\
  \hline

  & \\
  $1+r$ &  $2$ \\
  & \\
  \hline
  & \\
  $r+2 \leq w$ &  $\sum_{j=0}^{w-h-2-s} {n \choose w}{w \choose j} (-1)^j {k+w-n-j \brack s}$ \\
  & \\
 $ \leq n$ & $+ \sum_{j=w-r-h-1}^{w-h-1-s} 2  {n-r-1 \choose n-w+j-r-1} {n-w+j \choose j} (-1)^j ({k+w-n-j+1 \brack s}-{k+w-n-j \brack s})$ \\
 & \\
  \hline
\end{tabular}
\vspace{0.1in}
 \caption{Higher support weights of local MRC with two local codes.} \label{tab:hsw_data_two_local}

\end{table}

\end{thm}

\begin{proof}

We have to consider all sets $U$ such that $|U| = u \leq n-d_{s} = 2r +  2 - (h+1+s) = 2r-s-h+1$ (applying Proposition \ref{prop:ghw_matroid_local}). Note that  $2r-s-h+1 \geq r+1$, since  $1 \leq s \leq r-h$. The proof follows from the same counting as in the proof of Theorem \ref{thm:allsymbol_two_codes}.
\vspace{0.1in}
\ben
\item For a set $U$ of size $0 \leq u \leq r$, $\rho(U) = u$. The number of such sets are ${n \choose u}$. 
\vspace{0.05in}
\item For a set $U$ of size $r+1 \leq u \leq 2r-s-h+1$, there are the following two cases possible.   
\ben
\item The number of sets with $\rho(U) = u-1$ are $2{n-r-1 \choose u-r-1}$.
\vspace{0.05in}
\item The number of sets with $\rho(U) = u$ are ${n \choose u} - 2{n-r-1 \choose u-r-1}$. 
\een
\een
Hence, the $s^{th}$ support weight polynomial in \eqref{eq:higher_support} can be rewritten as
\begin{eqnarray}
W^{(s)}_C(Z) &= & \sum_{u=0}^r {n \choose u} Z^{n-u} (1-Z)^u {k-u \brack s} \nonumber \\
&& + \sum_{u=r+1}^{2r-s-h+1} 2{n-r-1 \choose u-r-1} Z^{n-u} (1-Z)^u {k-u+1 \brack s} \nonumber \\
&& + \sum_{u=r+1}^{2r-s-h+1} \left ( {n \choose u} - 2{n-r-1 \choose u-r-1} \right ) Z^{n-u} (1-Z)^u {k-u \brack s} \nonumber \\
& \stackrel{(a)}{=} & \sum_{u=0}^{2r-s-h} \sum_{j=0}^{u} {n \choose u} {u \choose j} Z^{n-u+j} (-1)^j {k-u \brack s} \nonumber \\
&& + \sum_{u=r+1}^{2r-s-h+1} \sum_{j=0}^{u} 2{n-r-1 \choose u-r-1} {u \choose j} Z^{n-u+j} (-1)^j ({k-u+1 \brack s} - {k-u \brack s}),
\end{eqnarray}
where $(a)$ follows since for the case when $u=2r-s-h+1$ and the set is full rank, the term in the summation is zero.
 By applying the change of variables $w = n-u+j, j=j$ to the above equation, the higher support weights given in the theorem statement follow.
 
%\begin{table}[ht]
%\centering
%\begin{tabular}{||c|c|c|l|c||}
%\hline
%\hline
%  Term & $w$ & No. of  & $(u,j)$ pairs & Range of $j$ \\
%  in Sum & & terms & & \\
%  \hline
%        
% \multirow{4}{*}{$T_1$} & $n$ & $2r-s-h$ & $(0,0), (1,1), \ldots, (2r-s-h,2r-s-h)$ & \\
%  \hhline{~---~}
%  & $n-1$ & $2r-s-h-1$ & $(1,0), (2,1), \ldots, (2r-s-h, 2r-s-h)$ & $j=0$ to\\
%  \hhline{~---~}
% & \vdots & \vdots & \vdots & $j=w-h-2-s$ \\
%  \hhline{~---~}
%& $h+2+s$ & $1$ & $(2r-s-h,0)$ & \\
%\hline
%\hline
%  \multirow{7}{*}{$T_2$} & $n$ & $r-s-h$ & $(r+1,r+1), (r+2,r+2), \ldots, (2r-s-h+1,2r-s-h+1)$  &  \\
%  \hhline{~---~}
%  & $n-1$ & $r-s-h$ & $(r+1,r), (r+2,r+1), \ldots, (2r-s-h+1, 2r-s-h)$ &  $j=w-r-1$ \\
%  \hhline{~---~}
%  & \vdots & \vdots & \vdots & to $j=w-h-1-s$ \\
%  \hhline{~---~}
%  & $r+2$ & $r-s-h$ &  $(r+1,1), (r+2,r-1), \ldots, (2r-s-h+1,r-s-h+1)$ &  \\
%  \hhline{~----}
%& $r+1$ & $r-s-h$ &  $(r+1,0), (r+2,1), \ldots, (2r-s-h+1,r-s-h)$ & $j=0$ to \\
%  \hhline{~---~}
% & \vdots & \vdots & \vdots & $j=w-h-1-s$ \\
%  \hhline{~---~}
%& $h+1+s$ & $1$ & $(2r-s-h+1,0)$ & \\
%        \hline
%      \hline
%
%\end{tabular}
%
%\vspace{0.1in}
% \caption{Applying the change of variable $w = n-u+j, j=j$.  Note that in the first term, the range of $j$ for a given $w$ is $0$ to $w-h-3$. 
% In the second term, for $h+2 \leq w \leq r+h+1$, the range of $j$ given $w$ is $0$ to $w-h-2$.  For $r+h+2 \leq w \leq n$, the range of $j$ given $w$ is $w-r-h-1$ to $w-h-2$.} \label{tab:change_variable_two}
%
%\end{table}

\end{proof}

\begin{thm} \label{thm:hsw_two_codes_local}
Consider an $[n,k,d_{\min}]$ local MRC with locality $r$, where $k+h = 2r$, $r+1 \leq k < 2r$, $n = k +h+ \frac{k+h}{r}$ and $d_{\min} = h+2$. For $r-h+1 \leq s \leq 2r-h$, the higher support weights of the code are given by
\begin{eqnarray}
A^{(s)}_w = {n \choose w} \sum_{j=0}^{w-n+2r-s}  {w \choose j} (-1)^j {k-n+w-j \brack s}_q, \ \  h+2+s \leq w \leq n.
\end{eqnarray}

\end{thm}

\begin{proof} We have to consider all sets $U$ such that $|U| = u \leq n-d_{s} = 2r  + 2 - (h+2+s) = 2r-h-s$ (applying Proposition \ref{prop:ghw_matroid_local}). Note that  $2r-h-s \leq r-1$, since  $r-h+1 \leq s \leq 2r-h$. 
The $s^{th}$ support weight polynomial in \eqref{eq:higher_support} can be rewritten as
\begin{eqnarray}
W^{(s)}_C(Z) &= & \sum_{u=0}^{2r-h-s} {n \choose u} Z^{n-u} (1-Z)^u {k-u \brack s} \nonumber \\
& = & \sum_{u=0}^{2r-h-s} \sum_{j=0}^{u} {n \choose u} {u \choose j} Z^{n-u+j} (-1)^j {k-u \brack s} \nonumber \\
& = & \sum_{w=n-2r+h+s}^{n} \sum_{j=0}^{w-n+2r-h-s} {n \choose w} {w \choose j} Z^{w} (-1)^j {k-n+w-j \brack s}.
\end{eqnarray}

\end{proof}

\section{Conclusions and Ongoing Work} \label{sec:concl}

 In this paper, we identified the matroid structures corresponding to data-local and local MRCs. The matroid structures of these codes are used to determine the associated Tutte polynomial. Using Greene's result, we obtained explicit expressions for the weight enumerators of data-local and local MRCs. Also using Britz's result, we obtained expressions for higher support weights of data-local and local MRCs with two local codes. It is part of ongoing work to investigate whether the non-negativity of these weight enumerators and higher support weights result in a non-trivial lower bound on the field size of MRC.
 
%\begin{proof}
 
%\end{proof}

\section*{Acknowledgements}

The first author would like to thank Prof. P. Vijay Kumar for the valuable discussions and comments.

\bibliographystyle{IEEEtran}
\bibliography{mrc}

\appendices

\section{Proof of Theorem \ref{thm:allsymbol_two_codes}} \label{app:allsymbol_two_codes}

For any set $U$ such that $|U| = u \leq n-d_{\min} =  k = 2r-h$, we have the following cases
\vspace{0.1in}
\ben
\item For a set $U$ of size $0 \leq u \leq r$, the set is independent. Hence, rank of the set $\rho(U) = u$. The number of sets  of this type are $N_{u,0} = {n \choose u}$. 
\vspace{0.05in}
\item For a set $U$ of size $r+1 \leq u \leq 2r-h$, there are the following two cases possible.   
\ben
\item Let $U$ be such that either $|U \cap S_1| = r+1$ or $|U \cap S_2| = r+1$. In this case, the set $U$ is dependent. The rank of the set $\rho(U)$, which is the size of the largest independent set in $U$, is given by $\rho(U) = u-1$. The number of sets $N_{u,1}$ of this type are $N_{u,1} = 2{n-r-1 \choose u-r-1}$.
\vspace{0.05in}
\item Let $U$ be such that $|U \cap S_i| \leq r, i = 1,2$. In this case, the set $U$ is independent and hence, rank of the set $\rho(U) = u$. The number of sets of this type are $N_{u,0}-N_{u,1} = {n \choose u} - 2{n-r-1 \choose u-r-1}$. 
\een
\een
Hence, the weight enumerator polynomial in \eqref{eq:weight_poly} can be rewritten as
\begin{eqnarray}
W_C(Z) &= &1+ \sum_{u=0}^r {n \choose u} Z^{n-u} (1-Z)^u (q^{k-u} - 1) \\
&& + \sum_{u=r+1}^{2r-h} 2{n-r-1 \choose u-r-1} Z^{n-u} (1-Z)^u (q^{k-u+1} - 1) \\
&& + \sum_{u=r+1}^{2r-h} \left ( {n \choose u} - 2{n-r-1 \choose u-r-1} \right ) Z^{n-u} (1-Z)^u (q^{k-u} - 1) \\
& = & 1+  \sum_{u=0}^{2r-h-1} {n \choose u} Z^{n-u} (1-Z)^u (q^{k-u} - 1) \\
&& + \sum_{u=r+1}^{2r-h} 2{n-r-1 \choose u-r-1} Z^{n-u} (1-Z)^u (q^{k-u+1} - q^{k-u}) \\
& = & 1+ \sum_{u=0}^{2r-h-1} \sum_{j=0}^{u} {n \choose u} {u \choose j} Z^{n-u+j} (-1)^j (q^{k-u} - 1) \\
&& + \sum_{u=r+1}^{2r-h} \sum_{j=0}^{u} 2{n-r-1 \choose u-r-1} {u \choose j} Z^{n-u+j} (-1)^j (q^{k-u+1} - q^{k-u}) \\
& = & 1+ T_1 + T_2.
\end{eqnarray}

 By applying the change of variables $w = n-u+j, j=j$ to the above equation, we first identify the $(u,j)$ pairs which result in a given value of $w$ (Table \ref{tab:change_variable_two_local}).

\begin{table}[ht]
\centering
\begin{tabular}{||c|c|c|l|c||}
\hline
\hline
  Term & $w$ & No. of  & $(u,j)$ pairs & Range of $j$ \\
  in Sum & & terms & & \\
  \hline
        
 \multirow{4}{*}{$T_1$} & $n$ & $2r-h$ & $(0,0), (1,1), \ldots, (2r-h-1,2r-h-1)$ & \\
  \hhline{~---~}
  & $n-1$ & $2r-h-1$ & $(1,0), (2,1), \ldots, (2r-h-1, 2r-h-2)$ & $j=0$ to\\
  \hhline{~---~}
 & \vdots & \vdots & \vdots & $j=w-h-3$ \\
  \hhline{~---~}
& $h+3$ & $1$ & $(2r-h-1,0)$ & \\
\hline
\hline
  \multirow{7}{*}{$T_2$} & $n$ & $r-h$ & $(r+1,r+1), (r+2,r+2), \ldots, (2r-h,2r-h)$  &  \\
  \hhline{~---~}
  & $n-1$ & $r-h$ & $(r+1,r), (r+2,r+1), \ldots, (2r-h, 2r-h-1)$ &  $j=w-r-1$ \\
  \hhline{~---~}
  & \vdots & \vdots & \vdots & to $j=w-h-2$ \\
  \hhline{~---~}
  & $r+2$ & $r-h$ &  $(r+1,1), (r+2,r-1), \ldots, (2r-h,r-h)$ &  \\
  \hhline{~----}
& $r+1$ & $r-h$ &  $(r+1,0), (r+2,1), \ldots, (2r-h,r-h-1)$ & $j=0$ to \\
  \hhline{~---~}
 & \vdots & \vdots & \vdots & $j=w-h-2$ \\
  \hhline{~---~}
& $h+2$ & $1$ & $(2r-h,0)$ & \\
        \hline
      \hline

\end{tabular}

\vspace{0.1in}
 \caption{Illustrating the change of variables from $u,j$ to $w=n-u+j, j$ for a local MRC with two local codes.} \label{tab:change_variable_two_local}

\end{table}

\end{document}